\documentclass[11pt]{article}
\usepackage{times}
\usepackage[dvips]{color}
\usepackage{amsmath}
\usepackage{amssymb}
\usepackage{amsthm}
\usepackage{epsfig}
\usepackage{url}

\usepackage{algorithm}
\usepackage{algorithmic}
\usepackage{verbatim}

\numberwithin{equation}{section}

\newtheorem{lemma}{Lemma}
\newtheorem{theorem}{Theorem}

\topmargin by -\headsep \textheight 8.9in \oddsidemargin 0pt
\evensidemargin \oddsidemargin \textwidth 6.5in
\floatstyle{ruled}
\newfloat{algo}{tbp}{lop}
\floatname{algo}{Algorithm}

\floatstyle{ruled}
\newfloat{algo}{tbp}{lop}
\floatname{algo}{Algorithm}

\usepackage{float}
\floatstyle{plain}\newfloat{myfig}{t}{figs}[section]
\floatname{myfig}{\textsc{Figure}}
\floatstyle{plain}\newfloat{myalg}{H}{algs}[section]
\floatname{myalg}{}
\begin{document}

\title{A Generalized Cheeger Inequality \thanks{Partially supported by the National Science Foundation under grant numbers CCF-1018463 and CCF-1149048.}
}
\author{Ioannis Koutis \\ {ioannis.koutis@upr.edu}  \and Gary Miller\\ {glmiller@cs.cmu.edu} \and Richard Peng\\ {rpeng@mit.edu} }

\maketitle

\begin{abstract}
\noindent The generalized conductance $\phi(G,H)$  between two graphs $G$ and $H$ on the same
vertex set $V$ is defined as
the ratio
$$
     \phi(G,H) = \min_{S\subseteq V} \frac{cap_G(S,\bar{S})}{ cap_H(S,\bar{S})},
$$
where $cap_G(S,\bar{S})$ is the total weight of the edges crossing from $S$ to $\bar{S}=V-S$.
We show that the minimum generalized eigenvalue
$\lambda(L_G,L_H)$ of the pair of Laplacians $L_G$ and $L_H$ satisfies
$$
     \lambda(L_G,L_H) \geq \phi(G,H) \phi(G)/8,
$$
where $\phi(G)$ is the usual conductance of $G$. A generalized cut that meets
this bound can be obtained from the generalized eigenvector
corresponding to $\lambda(L_G,L_H)$. The inequality complements a result
of Trevisan~\cite{Trevisan13} which shows that $\phi(G)$ cannot be replaced
by $\Theta(\phi(G,H))$ in the above inequality, unless the Unique Games Conjecture is false.
\end{abstract}

\section{The Cheeger inequality}

Let $G=(V,E,w)$ be a connected weighted graph. For $v\in V$ and $S\subseteq V$ we let
$$
   vol(v) = \sum_{(v,w)\in E} w(v,u) \textnormal{\qquad and \qquad}  vol(S) = \sum_{v \in S} vol(S).
$$

We will call $vol(v)$ the degree of vertex $v$.
We also denote by $cap(S,\bar{S})$ the total weight of edges
with exactly one endpoint in $S$ and one endpoint in $\bar{S}$.
The {\bf conductance} of $G$
is defined as
$$
   \phi(G) = \min_{S\subseteq V} \frac{cap(S,\bar{S})}{\min \{vol(S),vol(\bar{S})\}}.
$$

The Laplacian of $G$ is defined by
$$L(u,v) = -w(u,v) \textnormal{\qquad and \qquad} L(u,u) = \sum_{v\neq w} L(u,v).$$
The normalized Laplacian $\tilde{L}$ of $G$ is the matrix $D^{-1/2} L D^{-1/2}$ where $D$
is the diagonal of $L$. It is well understood that the normalized Laplacian is positive
semi-definite with a unique zero eigenvalue. If $\lambda_2$ is its second eigenvalue $\lambda_2$,
then the Cheeger inequality relates it to $\phi(G)$ as follows:
$$
    \lambda_2 \geq \phi(G)^2/2.
$$

At least one proof of the Cheeger inequality due to Mihail~\cite{DBLP:conf/focs/Mihail89} actually
shows something stronger, namely that for any
vector $y\bot Null(\tilde{L}_G)$, we can find a set $S_y$ such that
\begin{equation} \label{eq:anyvector}
     y^T \tilde{L}_G y > \left( \frac{cap_G(S_y,\bar{S_y})}{\min \{vol(S_y),vol(\bar{S_y})\}} \right)^2/2.
\end{equation}
The cut can be found by letting $S_y$ to consist
of the vertices corresponding to the $k$ smallest entries of $y$,
for some $1\leq k \leq n$.

\section{Generalized cuts for graph pairs}
We define the \textbf{generalized conductance} between $G$ and $H$ as follows:
$$
     \phi(G,H) = \min_{S\subseteq V} \frac{cap_G(S,\bar{S})}{ cap_H(S,\bar{S})}.
$$

To motivate this definition, we observe that the sparsest cut
problem can be captured within a factor of $2$ as a generalized
cut problem\footnote{sometimes called the non-uniform sparsest cut problem.} between two graphs. To this end let us
define the \textbf{demand graph}  $D_G=(V,E',w')$ with
every edge being present in $E'$ and the weights specified by
$$
   w'(u,v) = \frac{vol(u)vol(v)}{vol(V)}.
$$
Let $S\subseteq V$. Observe that by construction we have
$$
    cap_{D_G}(S,\bar{S}) = \frac{vol(S)vol(\bar{S})}{vol(V)}.
$$
Note now that
$$
    \min \{vol(S),vol(\bar{S})\} \geq \frac{vol(S)vol(\bar{S})}{vol(V)} \geq
    \min \{vol(S),vol(\bar{S})\}/2.
$$
From this it can be seen that
$$
     \frac{\phi(G)}{2} \leq  \phi(G,D_G)\leq \phi(G).
$$

A number of other problems can be viewed as generalized cut problems.

For example, consider the {\bf isoperimetric number} defined by:
$$
     h(G) = \min_{S\subseteq V} \frac{cap_G(S,\bar{S})}{\min \{|S|,|\bar{S}|\}}.
$$
If $K_n$ is the complete graph on $n$ vertices with edges weighted
by $1/n$, i.e. the identity over the space of sets orthogonal to
the constant vectors, it can be verified that we have
$$
   \frac{\phi(G,K_n)}{2} \leq h(G) \leq \phi(G,K_n).
$$

Another example is the min $s$-$t$ cut problem which looks for a cut of
minimum value among all possible cuts that separate $s$ and $t$. If
we denote that value by $\mu_{s,t}$, and let $G_{s,t}$ be the Laplacian of the edge $(s,t)$,
we have
$$
    \mu_{s,t} = \phi(G,G_{s,t}).
$$

\subsection{Cuts and Laplacians}

The value of a cut between $S$ and $\bar{S}$ can be expressed in terms of the
graph Laplacian as:
$$
       cap_G(S,\bar{S}) = x^T_S L_G x_S,
$$
where $x_S$ is characteristic vector of $S$, i.e.
the vector with ones in its entries
corresponding to $S$ and zeros in all other entries.

It follows that the generalized conductance can be expressed
as an optimization problem over the discrete 0-1 vectors:
$$
    \phi(G,H) = \min_{\substack{x \in \{0,1\}^n }} \frac{x^T L_G x}{x_T L_H x}.
$$
One can relax this discrete problem
over the real numbers:
$$
    \lambda(G,H) = \min_{\substack{x \in {\mathbb R}^n \\ x^Td =0}} \frac{x^T L_G x}{x_T L_H x}.
$$
Here $d$ is the vector containing the degrees of the vertices in $G$. The constraint $x^Td =0$
can be considered as a `normalization' constraint that fixes one representative for all vectors of the
form $x_c = y+c1$ that achieve the same ratio. That is, the constraint
doesn't change the minimum value of the ratio and we add it here because it will be useful in the sequel.
It is well known that $\lambda(G,H)$ is the first non-trivial generalized eigenvalue
for the problem $L_G x = \lambda L_H x$.

Note that the minimum eigenvalue of the normalized Laplacian can also be seen
as the first non-trivial generalized eigenvalue for the problem $L x= \lambda D x$
and thus as a continuous relaxation of the corresponding optimization problem.
We aim to prove a similar characterization for the generalized conductance of
any pair of graphs.

\section{Generalized Cheeger Inequality}

We begin with two Lemmas.

\begin{lemma} \label{th:splitting}
For all $a_i, b_i > 0$ we have
    $$\frac{\sum_i {a_i}}{\sum_i b_i} \geq \min_i \left\{ \frac{a_i}{b_i}\right\}.$$
\end{lemma}

\begin{lemma} \label{th:DvsDG}
Let $G$ be a graph, $d$ be the vector containing the degrees of the vertices,
and $D$ be corresponding diagonal matrix. For
all vectors $x$ where $x^Td = 0$ we have
$$
    x^T D x = x^T L_{D_G} x,
$$
where $D_G$ is the demand graph for $G$.
\end{lemma}
\begin{proof}
Let $d$ be the vector consisting of the entries
along the diagonal of $D$. By definition, we have
$$
   L_{D_G} = D - \frac{dd^T}{vol(V)}.
$$
The lemma follows.
\end{proof}

We prove the following theorem.
\begin{theorem}
\label{thm:generalizedcheeger}

Let $G$ and $H$ be any two weighted graphs and $D$ be
the vector containing the degrees of the vertices in $G$.
F any vector $x$ such that $x^Td =0$, we have

\[
\frac{x^TL_Gx}{x^TL_Hx} \geq \phi(G,D_G) \cdot \phi(G, H)/4,
\]
where $D_G$ is the demand graph of $G$
\end{theorem}

Let $V^{-}$ denote the set of $u$ such that $x_u \leq 0$ and
$V^{+}$ denote the set such that $x_u > 0$.
Then we can divide $E_G$ into two sets: $E_G^{same}$ consisting
of edges with both endpoints in $V^{-}$ or $V^{+}$, and
$E_G^{dif}$ consisting of edges with one endpoint in each.
In other words:
\begin{align*}
& E_G^{dif} = \delta_G\left(V^{-}, V^{+}\right),\text{ and}\\
& E_G^{same} = E_G \setminus E_G^{dif}.
\end{align*}

We also define $E_H^{dif}$ and $E_H^{same}$ similarly.

We first show a lemma which is identical to one used in the proof
of Cheeger's inequality~\cite{chung1}:

\begin{lemma} \label{lem:abssqr}
Let $G$ and $H$ be any two weighted graphs on the same vertex set $V$
partitioned into $V^{-}$ and $V^{+}$. For any vector $x$ we have
\[
\frac{ \sum_{uv \in E_G^{same}}w_G\left(u,v\right)\left|x_u^2-x_v^2\right| +
\sum_{uv \in E_G^{dif}}w_G(u, v) \left(x_u^2 + x_v^2\right) }{x^T L_H x}
\geq \frac{\phi(G, H)}{2} .
\]
\end{lemma}

\begin{proof}

We begin with a few algebraic identities:

Note that $2x_u^2+2x_v^2-(x_u - x_v)^2=(x_u+x_v)^2 \geq 0$
gives:
\[
\left(x_u - x_v\right)^2 \leq 2x_u ^2 + 2x_v^2.
\]

Also, suppose $uv \in E_H^{same}$ and without loss of generality that $|x_u| \geq |x_v|$.
Then letting $y = |x_u| - |x_v|$, we get:
\begin{eqnarray*}
|x_u^2-x_v^2| & = & \left(\left|x_v\right| + y\right) ^2 - \left|x_v\right|^2 \\
& = & y^2 + y |x_v|\\
& \geq & y^2 = \left(x_u - x_v\right)^2.
\end{eqnarray*}
The last equality follows because $x_u$ and $x_v$ have the same sign.

We then use the above inequalities to decompose the $x^T L_H x$ term.
\begin{eqnarray}
 x^T L_H   & = & \sum_{uv \in E_H^{same}}w_H(u,v)\left(x_u -x_v\right)^2
+ \sum_{uv \in E_H^{dif}}w_H(u, v)\left(x_u - x_v\right)^2 \nonumber \\
& \leq & \sum_{uv \in E_H^{same}} w_H(u, v) \left(x_u -x_v\right)^2
+ \sum_{uv \in E_H^{dif}} w_H(u, v) \left(2x_u^2 + 2x_v^2\right) \nonumber \\
& \leq & 2 \left( \sum_{uv \in E_H^{same}}  w_H(u, v) \left(x_u -x_v\right)^2
+ \sum_{uv \in E_H^{dif}} w_H(u, v) \left(x_u^2 + x_v^2\right) \right) \nonumber \\
& \leq & 2\left( \sum_{uv \in E_H^{same}}  w_H(u, v) \left|x_u^2 - x_v^2\right|
+ \sum_{uv \in E_H^{dif}} w_H(u, v) \left(x_u^2 + x_v^2\right) \right).
\end{eqnarray}

We can now decompose the summation further into parts for
$V^{-}$ and $V^{+}$:
\begin{align*}
& \sum_{uv \in E_G^{same}}w_G\left(u,v\right)\left|x_u^2-x_v^2\right| +
\sum_{uv \in E_G^{dif}}w_G\left(u, v\right) \left(x_u^2 + x_v^2\right) \\
= &\sum_{u \in V^{-}, v \in V^{-}}w_G\left(u,v\right)\left|x_u^2-x_v^2\right|
+ \sum_{u \in V^{-}, v \in V^{+}}w_G\left(u,v\right) x_u^2 \\
& + \sum_{u \in V^{+}, v \in V^{+}}w_G\left(u,v\right)\left|x_u^2-x_v^2\right|
+ \sum_{u \in V^{-}, v \in V^{+}}w_G\left(u,v\right) x_u^2.
\end{align*}

Doing the same for $ \sum_{uv \in E_H^{same}}  w_H(u, v) |x_u^2 - x_v^2|
+ \sum_{uv \in E_H^{dif}} w_H(u, v) (x_u^2 + x_v^2)$ we get:
\begin{align*}
& \frac{ \sum_{uv \in E_G^{same}}w_G(u,v)\left|x_u^2-x_v^2\right| +
\sum_{uv \in E_G^{dif}}w_G(u, v) \left(x_u^2 + x_v^2\right) }{x^T L_H x} \\
\geq & \min \left\{
\frac{\sum_{u \in V^{-}, v \in V^{-}}w_G(u,v)\left|x_u^2-x_v^2\right|
+ \sum_{u \in V^{-}, v \in V^{+}}w_G(u,v) x_u^2}
{\sum_{u \in V^{-}, v \in V^{-}}w_H(u,v)\left|x_u^2-x_v^2\right|
+ \sum_{u \in V^{-}, v \in V^{+}}w_H(u,v) x_u^2},\right.\\
&\qquad \left.\frac{\sum_{u \in V^{+}, v \in V^{+}}w_G(u,v)\left|x_u^2-x_v^2\right|
+ \sum_{u \in V^{-}, v \in V^{+}}w_G(u,v) x_v^2}
{\sum_{u \in V^{+}, v \in V^{+}}w_H(u,v)\left|x_u^2-x_v^2\right|
+ \sum_{u \in V^{-}, v \in V^{+}}w_H(u,v) x_v^2}
\right\}.
\end{align*}
The inequality comes from applying of Lemma \ref{th:splitting}.

By symmetry in $V^{-}$ and $V^{+}$, it suffices to show that
\begin{equation} \label{eq:toprove}
 \frac{\sum_{u\in V^{-}, v\in V^{-}}w_G\left(u, v\right)\left|x_u^2 -x_v^2\right| +
 \sum_{u\in V^{-}, v\in V^{+} }w_G(u, v) x_u^2}{
 \sum_{u\in V^{-}, v\in V^{-}}w_G\left(u, v\right)\left|x_u^2-x_v^2\right|
+ \sum_{u \in V^{-}, v \in V^{+}}w_G\left(u, v\right) x_u^2 }
\geq \phi(G, H).  \\
\end{equation}

We sort the $x_u$ in increasing order of $|x_u|$ into
such that $x_{u_1} \geq \ldots  \geq  x_{u_k}$, and let $S_k=\{x_{u_1},\ldots,x_{u_k}\}$. We have
\begin{align*}
& \sum_{u\in V^{-}, v\in V^{-}}w_G(u, v)\left|x_u^2 -x_v^2\right| +
 \sum_{u\in V^{-}, v\in V^{+} }w_G(u, v) x_u^2
=  \sum_{i=1 \dots k} \left(x_{u_i}^2 - x_{u_{i-1}}^2\right) cap_G\left(S_k, \bar{S_k}\right),
\end{align*}
and
\begin{align*}
& \sum_{u\in V^{-}, v\in V^{-}}w_H(u, v)\left|x_u^2 -x_v^2\right| +
 \sum_{u\in V^{-}, v\in V^{+} }w_H(u, v) x_u^2
 = \sum_{i=1 \dots k} \left(x_{u_i}^2 - x_{u_{i-1}}^2\right) cap_H\left(S_k, \bar{S_k}\right).
\end{align*}

Applying Lemma \ref{th:splitting} we have
\[
 \frac{\sum_{u\in V^{-}, v\in V^{-}}w_G(u, v)|x_u^2 -x_v^2| +
 \sum_{u\in V^{-}, v\in V^{+} }w_G(u, v) x_u^2}
 { \sum_{u\in V^{-}, v\in V^{-}}w_G\left(u, v\right)\left|x_u^2-x_v^2\right|
+ \sum_{u \in V^{-}, v \in V^{+}}w_G\left(u, v\right) x_u^2 }
\geq \min_k \frac{cap_H\left(S_G, \bar{S_i}\right)}{cap_H\left(S_i, \bar{S_i}\right)}
\geq \phi(G, H),
\]
where the second inequality is by definition of $\phi(G,H)$. This proves equation \ref{eq:toprove} and the Lemma follows.
\end{proof}

We now proceed with the proof of the main Theorem.

\begin{proof}

We have
\begin{eqnarray}
{x^TL_Gx} & = & \sum_{uv \in E_G}w_G(u,v)(x_u-x_v)^2 \nonumber\\
& = & \sum_{uv \in E_G^{same}}w_G(u,v)(x_u-x_v)^2
+ \sum_{uv \in E_G^{dif}}w_G(u,v)(x_u-x_v)^2 \nonumber \\
& \geq & \sum_{uv \in E_G^{same}}w_G(u,v)(x_u-x_v)^2
+ \sum_{uv \in E_G^{dif} }w_G(u,v) (x_u^2 + x_v^2). \nonumber \\
\end{eqnarray}
The last inequality follows by $x_ux_v \leq 0$ as $x_u \leq 0$ for all
$u \in V^{-}$ and $x_v \geq 0$ for all $v \in V^{+}$.

We multiply both sides of the inequality by
$$\sum_{uv \in E_G^{same}} w_G(u,v) (x_u + x_v)^2
 + \sum_{uv \in E_G^{dif}} w_G(u, v) (x_u^2 + x_v^2).$$

We have
\begin{eqnarray*}
&\left( \sum_{uv \in E_G^{same}}w_G(u,v)(x_u-x_v)^2
+ \sum_{uv \in E_G^{dif} }w_G(u,v) (x_u^2 + x_v^2)  \right)\\
& \cdot \left( \sum_{uv \in E_G^{same}} w_G(u,v) (x_u + x_v)^2
 + \sum_{uv \in E_G^{dif}} w_G(u, v) (x_u^2 + x_v^2) \right) \\
\geq & \left(  \sum_{uv \in E_G^{same}} |x_u - x_v| |x_u + x_v|
+ \sum_{uv \in E_G^{dif}} w_G(u, v) (x_u^2 + x_v^2) \right)^2\\
= & \left(  \sum_{uv \in E_G^{same}} |x_u^2 - x_v^2|
+ \sum_{uv \in E_G^{dif}} w_G(u, v) (x_u^2 + x_v^2) \right)^2.
\end{eqnarray*}

Furthermore, notice that $(x_u + x_v)^2 \leq 2x_u ^2 + 2x_v ^2$ since
$2x_u^2 + 2x_v^2 - (x_u + x_v)^2 = (x_u - x_v)^2 \geq 0$.
So, we have
\begin{align*}
& \sum_{uv \in E_G^{same}} w_G(u,v) (x_u + x_v)^2
 + \sum_{uv \in E_G^{dif}} w_G(u, v) (x_u^2 + x_v^2) \\
\leq & 2 \left( \sum_{uv \in E_G^{same}} w_G(u,v) (x_u^2 + x_v^2)
 + \sum_{uv \in E_G^{dif}} w_G(u, v) (x_u^2 + x_v^2) \right)\\
 & = 2 x^T Dx \leq 4 x^T L_{D_G} x,
\end{align*}
where $D$ is the diagonal of $L_G$ and the last inequality
comes from Lemma~\ref{th:DvsDG}. Combining the last two inequalities we get:
\begin{eqnarray*}
\frac{x^TL_Gx}{x^TL_Hx} \geq & \frac{1}{2}
\cdot \left( \frac{\sum_{uv \in E_G^{same}} \left|x_u^2 - x_v^2\right|
+ \sum_{uv \in E_G^{dif}} w_G(u, v) \left(x_u^2 + x_v^2\right)}{x^T L_H x} \right)\\
& \cdot \left( \frac{\sum_{uv \in E_G^{same}} \left|x_u^2 - x_v^2\right|
+ \sum_{uv \in E_G^{dif}} w_G(u, v) \left(x_u^2 + x_v^2\right)}{x^T L_{D_G} x} \right).
\end{eqnarray*}

By Lemma \ref{lem:abssqr}, we have that the first factor is bounded by
$\frac{1}{2} \phi(G, H)$  and the second factor bounded by
$\frac{1}{2} \phi(G, D_G)$. Hence we get
\begin{align}
\frac{x^TL_Gx}{x^TL_Hx}
& \geq \frac{1}{4} \phi(G, H) \phi(G,{D_G}).
\end{align}
\end{proof}

\section{Computation}

Note that in Lemma \ref{lem:abssqr} we actually proved that for any vector $x$ such that $x^Td =0$,
we can sort $x$ and find $n-1$ sets $S_i\subseteq V$ such that:
$$
    \min_i \frac{ cap_G(S_i,\bar{S_i})}{cap_H(S_i,\bar{S_i})} \leq \frac{1}{{\phi(G,D_G)}}\cdot \frac{x^T L_G x}{x^T L_H x}.
$$
This suggests that we can find a cut which is at most $1/\phi(G,D_G)$ larger than the ratio $(x^TL_Gx/x^TL_Hx).$

Given any positive definite matrix $A$, one can use the inverse power iteration $y_{i+i} = A^{-1} y_i$, where $y_0$ is a random vector, to find a vector $x$ such that
\begin{equation} \label{eq:approxvector}
     \frac{x^T A x}{x^T x} \leq (1+\epsilon) \lambda_{\min}(A).
\end{equation}
The number of rounds required for this is $O(\log n / \epsilon)$; for a proof  see~\cite{spielman-2006}.

Analogously, given a pair of positive definite matrices $(A,B)$, one can perform power iteration with the  matrix $A^{-1}B$ to find a vector $x$ such that
$$
   \frac{x^T A x}{x^T B x} \leq (1+\epsilon) \lambda_{\min}(A,B).
$$
The proof is similar to the simple eigenvalue problem case, using only the additional fact that the generalized eigenvectors of the pair $(A^{-1},B^{-1})$ are the usual eigenvectors of the matrix $A^{-1}B$, in addition with the fact
that the eigenvectors are $A$-orthogonal and $B$-orthogonal~\cite{Stewart.Sun}. Note that the iteration can be extended to the case when $A$ has a known null space, by simply operating on vectors orthogonal to the null space.

Note now that $A^{-1}By_i$ can be implemented as a linear system solve $Az = By_i$. Instead of solving exactly a linear system with the Laplacian, one can use a more efficient iterative solver instead, and ask for a solution $\tilde{z}$ that satisfies $||\tilde{z}-z||_{A} \leq (1+\epsilon/4) ||A^{-1}y_i||_{A}$. For  many iterative linear system solvers, including the fast Laplacian solvers, this step is an implicit multiplication with a matrix $\tilde{A}^{-1}$ which is spectrally close to $A^{-1}$. Spielman and Teng~\cite{spielman-2006} observe that this is sufficient for the computation of an approximate eigenvector that satisfies inequality~\ref{eq:approxvector}.This extends to the generalized problem with Laplacians as well. Finally, a little more care has to be taken for the case of Laplacian solvers that are randomized. In that case $O(\log(1/p))$ different runs of the inverse power method are need to get a good approximate eigenvector with probability at least~$1-p$.

\bibliographystyle{alpha}

\end{document}